\documentclass[prl,twocolumn,aps,superscriptaddress]{revtex4}

\usepackage{graphicx}
\usepackage{dcolumn}
\usepackage{bm}
\usepackage{amsthm}
\usepackage{dsfont}
\usepackage{array}
\usepackage{amsmath}
\usepackage{mathbbol}
\usepackage{epic}
\usepackage{calc}
\usepackage{amsfonts}
\usepackage{amssymb}

\newtheorem*{theorem*}{Theorem}

\newtheorem*{lemma*}{Lemma}

\newtheorem{propo}{Proposition}

\DeclareMathOperator{\Tr}{Tr}
\providecommand{\openone}{\leavevmode\hbox{\small1\kern-3.8pt\normalsize1}}

\begin{document}
\title{Steering Criteria via Covariance Matrices of Local Observables in Arbitrary Dimensional Quantum Systems}
\author{Se-Wan Ji, Jaehak Lee, Jiyong Park and Hyunchul Nha }
 \affiliation{Department of Physics, Texas A$\&$M University at Qatar University, PO Box 23784, Doha, Qatar}

\begin{abstract}
We derive steerability criteria applicable for both finite and infinite dimensional quantum systems using covariance matrices of local observables. We show that these criteria are useful to detect a wide range of entangled states particularly in high dimensional systems and that the Gaussian steering criteria for general $M\times N$-modes of continuous variables are obtained as a special case. Extending from the approach of entanglement detection via covariance matrices, our criteria are based on the local uncertainty principles incorporating the asymmetric nature of steering scenario.  Specifically, we apply the formulation to the case of local orthogonal observables and obtain some useful criteria that can be straightforwardly computable, and testable in experiment, with no need for numerical optimization. 
\end{abstract} 
\maketitle

In quantum world, there exist some strong correlations that cannot be described in classical ways providing thereby a crucial basis for applications, e.g. in quantum information processing. Among different forms of quantum correlations, the most well studied are quantum entanglement \cite{Horodecki} and nonlocality \cite{Brunner}. Nonlocality is the strongest correlation that does not admit any local realistic models \cite{Bell}, in which  
the joint probability for the outcomes $a$ and $b$ of local measurements $A$ and $B$, respectively, are explained by 
\begin{eqnarray}
P_{\rm LHV}(a,b|A,B)=\sum_\lambda p_\lambda P_\lambda(a|A)P_\lambda(b|B),
\end{eqnarray}
where a hidden-variable $\lambda$ is chosen according to the distribution $p_\lambda$. 
On the other hand, quantum entanglement is the correlation distinguished from classical correlation {\it within the framework of quantum mechanics}. That is, if a quantum state shows correlation that cannot be explained by the form 
\begin{eqnarray}
P_Q(a,b|A,B)=\sum_\lambda p_\lambda P_\lambda^Q(a|A)P_\lambda^Q(b|B),
\end{eqnarray}
where the superscript $Q$ refers to the restriction to quantum statistics only, it is called quantum entangled.

Recently, an intermediate form of correlation between quantum entanglement and nonlocality was rigorously defined in \cite{Wiseman07}---quantum steering---and it has attracted a great deal of interest during the past decade. The concept of quantum steering envisions a situation where Alice performs a local measurement on her system, which makes it possible to steer Bob's local state depending on her choice of measurement setting \cite{EPR,Sch}. This notion is practically relevant when Bob wants to confirm quantum correlation although he cannot trust Alice or her devices at all \cite{Wiseman07}, leading to some applications, e.g. one-sided device-independent cryptography \cite{Branciard} and sub-channel discrimination \cite{Piani15}. In view of joint probability distribution, steering is the quantum correlation that can rule out the local hidden state (LHS) models,  
\begin{eqnarray}
P_{\rm LHS}(a,b|A,B)=\sum_\lambda p_\lambda P_\lambda(a|A)P_\lambda^Q(b|B),
\end{eqnarray}  
where Alice's statistics $P_\lambda(a|A)$ is unrestricted while Bob' statistics $P_\lambda^Q(b|B)$ obeys quantum principles. $P_{\rm LHS}$ is obviously a subset of $P_{\rm LHV}$ as seen from its construction, which makes EPR steering more accessible in experiment than nonlocality \cite{Smith}. 
There have been other remarkable works on quantum steering including its connection to measurement incompatibility \cite{Quintino} and the phenomenon of one-way steering  \cite{Brunner14,Handchen}, etc.. However, we need to have a more comprehensive set of steering criteria readily testable particularly for higher-dimensional systems, which may bring us a deeper understanding of quantum correlation. 

In this work, we introduce steering criteria based on covariance matrices of local observables that can be applied to bipartite quantum systems of arbitrary dimensions. This approach is an extension from the entanglement detection via covariance matrices \cite{Guhne07,Gitt,Guhne} by incorporating the asymmetric nature of quantum steering and local uncertainty relations. 
In particular, we apply our formalism to local orthogonal observables and derive some useful criteria that can be readily computable, and also practically testable, without doing numerical optimizations. We illustrate the usefulness of our methods by detecting steerability of some higher-dimensional states, for which few criteria are known so far. Moreover, we show that our method leads to the Gaussian steering criteria for general $M\times N$-modes of continuous variables (CVs) as a special case \cite{Wiseman07, Adesso15, JKN15}.


{\it Non-steerability}---Let us begin with the notion of non-steerablity. Assume that two separate observers, Alice and Bob, share a bipartite quantum state $\rho_{AB}$ on $ {\mathcal H} = {\mathcal H}_A \otimes {\mathcal H}_B$, where $d_A$ $\left( d_B \right) $ is the dimension of $ {\mathcal H}_A$ $\left( {\mathcal H}_B \right)$. If this state is nonsteerable from Alice to Bob, then the joint probabilities of local measurements on two observers can be written as \cite{Wiseman07}
\begin{equation}
\label{nonsteer}
\begin{array}{lcl}
P \left( a, b | A_k, B_{l} \right) & = & Tr\left[ \rho_{AB} A_k(a) \otimes B_l(b) \right] \\
\\
 & = & \sum_{\lambda} P(\lambda) P(a| A_k, \lambda ) Tr\left[ \rho_{\lambda} B_l(b) \right],
\end{array}
\end{equation} 
where Alice's and Bob's POVMs are denoted by $ \left\{ A_k \right\} $ and $ \left\{ B_l \right\} $ respectively.  In this case, the (un-normalized) conditional state held by Bob, when Alice performs the measurement $ A_{k} $ with the outcome $a$, is given by
\begin{equation}
\label{assemblage}
\rho_{a|k}=\sum_{\lambda}P\left( \lambda \right)P\left(a|A_k,\lambda \right) \rho_{\lambda}.
\end{equation}  
The probability of Alice getting the output $a$ for measurement $A_k$ is given by $ Tr\left( \rho_{a|k} \right) =\sum_{\lambda} P(\lambda)P\left(a|A_k,\lambda\right)$, and the set of un-normalized states $\left\{ \rho_{a|k} \right\} $ is referred to as an assemblage \cite{Piani15}. If the assemblage for a given state $\rho_{AB}$ is written in the form of Eq. (\ref{assemblage}) for all measurements $ \left\{ A_k \right\}$ and outcomes $a$, we say that the correlation of the given state $\rho_{AB}$ can be explained by the local hidden state (LHS) model \cite{Wiseman07}.   \\
\\
{\it Local uncertainty relations}---Let $\left\{ \hat{A}_k\right\}$ be the observables on ${\mathcal H}_A$. If they do not have a common eigenstate, there exists a nontrivial, state-independent, bound $C_A > 0$ such that $ \sum_k \delta^2 \left( \hat{A}_k\right)_{\rho_A} =\sum_k \langle \hat{A}_k^2 \rangle -\langle \hat{A}_k \rangle^2 \geq C_A $ holds for all states $\rho_A$ on ${\mathcal H}_A$ \cite{HofTak03}, which is the so-called local uncertainty relation (LUR). For example, if we consider the three Pauli operators $ \sigma_x,\sigma_y,\sigma_z$ in a qubit system, then $C_A =2$. Another example is the case of two observables $ X^{(N)}=\frac{1}{\sqrt{2}}\left( a^N +a^{\dagger N} \right), P^{(N)}=\frac{-i}{\sqrt{2}}\left( a^N-a^{\dagger N} \right)$ in continuous variable systems, for which we have $C_A =N!$ due to $ \left[ X^{(N)}, P^{(N)} \right] = i\sum_{r=1}^{N}r! \left( \begin{array}{cc} N \\
r \end{array} \right)^2 a^{\dagger N-r}a^{N-r}$ \cite{JKN15, Hillery87}. The simplest case of $N=1$ for the operators $ X^{(N)}, P^{(N)}$ correspond to two orthogonal quadrature amplitudes. \\

The following is a non-steerability criterion that is based on LURs. 
\begin{lemma*}
(Steering criteria with local uncertainty relations) - If a given bipartite quantum state $\rho_{AB}$ satisfies Eqs. (\ref{nonsteer}) and (\ref{assemblage}), i.e., nonsteerable from Alice to Bob (from Bob to Alice), the following inequality must be satisfied 
\begin{equation}
\label{LUR0}
\sum_k^N \delta^{2} \left( \hat{A}_k \otimes \openone + \openone \otimes \hat{B}_k \right) \geq C_B\; \left(C_A\right), 
\end{equation}
where $C_B$ is a strict positive lower bound of LURs \cite{HofTak03} in Bob's Hilbert space, i.e.,
\begin{equation}
\label{BLUR}
\sum_{k} \delta^{2} \left( \hat{B}_k \right) \geq C_B > 0.
\end{equation}
\end{lemma*}
\begin{proof}
We consider a situation where Alice attempts to infer Bob's measurement outcome by performing a measurement on her subsystem. Let $\hat{B}_k^{e}\left(\hat{A}_k \right)$ be Alice's estimate of the outcome of Bob's measurement $\hat{B}_k$ as a function of the outcome of her measurement $\hat{A}_k$. The average inferred variance of $\hat{B}_k$ given estimate $B_k^{e}\left( \hat{A}_k \right)$ is defined as 
\begin{equation}
\label{inferred}
\delta_{inf}^{2}\left(\hat{B}_k \right)=\langle \left[ \hat{B}_k -\hat{B}_k^e \left( \hat{A}_k \right) \right]^2 \rangle.
\end{equation}
Here the average is taken over all possible outcomes $a_k$ and $b_k$. It can be readily shown that if a given quantum state $\rho_{AB}$ is nonsteerable from Alice to Bob, then $\sum_k \delta_{inf}^2\left(\hat{B_k} \right) \geq C_B$ following the method in \cite{reid2, Cavalcanti09}.  With a linear estimate $ \hat{B}_k^e \left( \hat{A}_k \right)= -g_k \hat{A}_k + \langle \hat{B}_k + g_k \hat{A}_k \rangle$ \cite{reid2, Cavalcanti09}, where $\left\{g_k\right\}$ are arbitrary real numbers, we have
\begin{equation}
\label{gLUR}
\sum_k \delta^2\left( g_k \hat{A}_k \otimes \openone + \openone \otimes \hat{B}_k \right) \geq C_B.
\end{equation}
Then if we set $g_1=...=g_N=1$, we obtain the desired inequality in Eq. (\ref{LUR0}).
\end{proof}

{\it Covariance Matrix}---Let $\rho$ be a given quantum state and let $ \left\{ O_k\; : k=1,...,N \right\}$ be some observables. Then the elements of $ N \times N$ symmetric covariance matrix $\gamma$ are defined by
\begin{equation}
\label{CM}
\gamma_{i,j}=\left( \langle O_i O_j \rangle + \langle O_j O_i \rangle \right)/2 -\langle O_i \rangle \langle O_j \rangle.
\end{equation}
\\
 Now, let us consider a total set of observables in a composite system, $\left\{ O_k \right\}=\left\{ A_k \otimes \openone, \openone \otimes B_k \right\}$ to construct the covariance matrix. Then the covariance matrix $ \gamma_{AB}$ with $\left\{ O_k \right\}$  has the block form
\begin{equation}
\label{BCM}
\gamma_{AB} \left( \rho_{AB},\, \left\{O_k\right\} \right)=\left[ \begin{array}{cc}
A & C \\
C^{T} & B \\
\end{array} \right],
\end{equation}
where $A=\gamma\left( \rho_A, \left\{ A_k \right\} \right)$ and $B= \gamma \left( \rho_B, \left\{ B_k \right\} \right) $ are covariance matrices for the reduced states $\rho_A$ and $\rho_B$, respectively, and  the correlation matrix $C$ has the entries $C_{k,l}=\langle A_k \otimes B_l \rangle -\langle A_k \rangle \langle B_l \rangle $. 

\begin{theorem*}
If a given bipartite quantum state $ \rho_{AB}$ is non-steerable from Alice to Bob, its covariance matrix $ \gamma_{AB}$ satisfies
\begin{equation}
\label{covariance}
\gamma_{AB} \geq { {\bf 0}_{A}} \oplus \kappa_B,
\end{equation} 
where $\kappa_B= \sum_k p_k \gamma \left( |b_k\rangle \langle b_k| \right)$, $|b_k\rangle$ are the states on ${\mathcal H}_B$, and $\sum_k p_k |b_k \rangle\langle b_k | =\rho_B$ is Bob's reduced state. 
\end{theorem*}
\begin{proof}
To prove Theorem, we make use of the techniques that were utilized in Ref.\cite{Guhne07} for entanglement detection.\\
Let us define a set of matrices as $T:=\left\{ t \; | \; t={\bf 0}_A \oplus \kappa_B + P \;\; with \;\; P \geq 0 \right\}$, which forms a closed convex cone. Then {\bf Theorem} is reformulated by saying that if $\rho_{AB}$ is non-steerable from Alice to Bob, then $ \gamma_{AB} \in T$. If a given $\rho_{AB}$ violates the inequality in Eq. (\ref{covariance}), we have $ \gamma_{AB} \notin T$. \\

As indicated by a corollary to the Hanh-Banach theorem, for each $ \gamma_{AB} \notin T$, there exists a symmetric matrix $ W $ and a real number $R$ such that $\Tr \left(W \gamma_{AB} \right) < R$ while
\begin{equation}
\label{Hahn}
\Tr \left( W t \right) > R, \;\; \forall t \in T.
\end{equation}  
Because $ \Tr \left( W P\right) \geq 0$ holds for all $P \geq 0$, we have $ W \geq 0$. Let us use the spectral decomposition of $W=\sum_{k} \lambda_k \vec{\Gamma}^{(k)} \vec{\Gamma}^{(k)T}  \equiv \sum_k \lambda_k \overrightarrow{ \left( \alpha^{(k)} \oplus \beta^{k}\right)} \overrightarrow{ \left(\alpha^{(k)} \oplus \beta^{(k)} \right)}^T$. Introducing $\hat{A}_k=\sqrt{ \lambda_{(k)} }\sum_{i} \alpha^{(k)}_{i} A_i$ and $\hat{B}_{k}=\sqrt{ \lambda_k }\sum_{i} \beta^{(k)}_{i} B_i$ where $ \left\{ A_i \right\}$ and $\left\{ B_i \right\}$ are the observables chosen for the construction of $\gamma_{AB}$, we have for $\rho_{AB}$ that
\begin{equation}
\label{LUR1}
\Tr \left( W \gamma_{AB} \right)=\sum_{k} \delta^2\left( \hat{A}_k \otimes \openone + \openone \otimes \hat{B}_k \right).
\end{equation}
By definition we know that all $ {\bf 0}_A \oplus \kappa_B \in T$, and by the concavity of covariance matrix it follows that all $ {\bf 0}_A \oplus \gamma_B  \in T$.  Thus $ \Tr \left( W \left( {\bf 0}_A \oplus \gamma_B \right) \right) =\sum_k \delta^2 \left( \hat{B}_k \right) >R $. This implies that
\begin{equation}
\label{LUR2}
R\;<\; \min_{\rho_B} \left( \sum_k \delta^2 \left( \hat{B}_k \right)_{\rho_B} \right) = C_B.
\end{equation}
Eventually, since the inequality in Eq. (\ref{covariance}) is violated, $ \gamma_{AB} \notin T$ and $ \sum_k \delta^2 \left( \hat{A}_k \otimes \openone + \openone \otimes \hat{B}_k \right) = \Tr \left( W \; \gamma_{AB} \right)\; < R \;< C_B$, showing a violation of inequality in Eq. (\ref{LUR0}). \\

This means that if a given quantum state $ \rho_{AB}$ violates the inequality in Eq. (\ref{covariance}), there must exist the sets $ \left\{ \hat{A}_k \right\}$, $ \left\{ \hat{B}_k \right\}$ for which the given state violates the inequality in Eq. (\ref{LUR0}) and then $\rho_{AB}$ is steerable. Therefore we conclude that if a given quantum state $\rho_{AB}$ is nonsteerable from Alice to Bob, its covariance matrix $\gamma_{AB}$ must satisfy the inequality in Eq. (\ref{covariance}).  
\end{proof}
In contrast to the non-steerability condition in Eq. (\ref{covariance}), we note that the separability condition reads as $\gamma_{AB} \geq {\kappa}_{A} \oplus \kappa_B$ \cite{Guhne07,Gitt}, where the local covariance matrix ${\kappa}_{A}$ appears due to the restriction to quantum statistics at Alice's station as well.

{\it Local Orthogonal Observables}---We now derive some readily computable steering criteria using {\bf Theorem}. 
For this purpose, let us choose $d_A^2$ observables $\left\{ \hat{A}_k \right\}$ on ${\mathcal H}_A$ such that they satisfy orthogonal relations $ \Tr{\left( \hat{A}_k\,\hat{A}_l \right)}=\delta_{kl}$.  Note that a quantum state $\rho_A$ can then be represented using these observables as $\rho_A=\sum_{k=1}^{d_A^2}\langle\hat{A}_k\rangle_{\rho_A}\hat{A}_k$. Similarly, we take the local observables $\left\{ \hat{B}_k \right\}$ in ${\mathcal H}_B$. These orthogonal observables are called local orthogonal observables (LOOs) \cite{Yu05}. In this case, the lower bound of the inequality (\ref{LUR0}) is given by $C_B=d_B -1$ \cite{Guhne06}. The covariance matrix $ \gamma_{AB}^{LOOs} $ is now constructed with the LOOs, and their partioned blocks are represented by $A^{LOOs}$ and $B^{LOOs}$ for the reduced states, respectively, and $C^{LOOs}$ for the correlation matrix. We then obtain the following result.
  
\begin{propo}
If a given bipartite quantum state $\rho_{AB}$ is nonsteerable from Alice to Bob, the correlation matrix $C^{LOOs}$ constructed with LOOs must satisfy
\begin{equation}
\label{correlation1}
\| C^{LOOs}\|_{tr} \leq \sqrt{\left(d_A -\Tr{\left(\rho_A^2\right)}\right) \left( 1-\Tr{\left( \rho_B^2 \right)}\right)}, 
\end{equation}   
where $ \| A \|_{tr} $ is a trace norm of a matrix $A$.
\end{propo}
\begin{proof}
If a partitioned matrix in block form is positive semidefinite, 
\begin{equation}
\label{block}
\Lambda=\left( \begin{array}{cc}
\Lambda_{11} & \Lambda_{12} \\
\Lambda_{12}^T & \Lambda_{22} \end{array} \right) \geq 0,
\end{equation}
we have the relation $\|\Lambda_{12} \|_{tr}^2 \leq \|\Lambda_{11} \|_{tr}\,\|\Lambda_{22}\|_{tr}$ \cite{Horn}. We thus find, due to the inequality (\ref{covariance}), 
\begin{equation}
 \begin{array} {ll}
\| C^{LOOs}\|_{tr} & \leq \sqrt{ \|A^{LOOs}\|_{tr} \|B^{LOOs}-\kappa_B^{LOOs} \|_{tr}  } \\
 & =\sqrt{\left(d_A -\Tr{\left(\rho_A^2\right)}\right) \left( 1-\Tr{\left( \rho_B^2 \right)}\right)}, \end{array} 
\end{equation}
where we used $ \|A^{LOOs}\|_{tr}=\Tr{A}=\sum_k^{d_A^2}\langle \hat{A}_k^2 \rangle -\langle \hat{A}_k \rangle^2  =d_A-\Tr{\left(\rho_A^2\right)}$ and $ \| B^{LOOs}-\kappa_B^{LOOs} \|_{tr} =\Tr{B^{LOOs}}-\Tr{\kappa_B^{LOOs}}=d_B-\Tr{\left(\rho_B^2\right)}-\left(d_B -1\right)=1-\Tr{\left(\rho_B^2\right)}$ \cite{Guhne07}. If a given state is nonsteerable from Bob to Alice, the upper bound in Eqs. (13) and (15) is $\sqrt{\left(d_B -\Tr{\left(\rho_B^2\right)}\right) \left( 1-\Tr{\left( \rho_A^2 \right)}\right)}$. We can also prove {\bf Proposition 1} without resort to the covariance matrix formalism in Appendix.  
\end{proof}
{\it Examples} - (i) Let us consider a noisy, {\it asymmetric}, two-qubit state that is written as
\begin{equation}
\label{nsinglet}
\rho_{AB}^{ns}=p\left| \psi^{-} \rangle \langle \psi^{-} \right| +\left(1-p \right) \rho_s,
\end{equation}
where $ | \psi^{-} \rangle = \frac{1}{\sqrt{2}} \left( |01 \rangle - |10 \rangle \right) $ and $ \rho_s = 2/3 |00 \rangle \langle 00| + 1/3 |01 \rangle \langle 01|$. Using the condition in Eq. (\ref{correlation1}), we find that the state in Eq. (\ref{nsinglet}) is steerable  from Alice to Bob  if $p  \gtrsim 0.53197 $, and steerable from Bob to Alice when $ p \gtrsim 0.53524 $. \\


(ii) Suppose now that Alice and Bob share a single pair of $3$-dimensional particles in the state
\begin{equation}
\label{statef}
\rho_{AB}^{F}=F | \Phi^{+} \rangle \langle \Phi^{+} | + \frac{1-F}{3}\left( |01\rangle \langle 01| + |12 \rangle \langle 12| + |20 \rangle \langle 20 | \right),
\end{equation}
where $ |\Phi^{+} \rangle =1/\sqrt{3} \left( |00\rangle + |11 \rangle +|22 \rangle \right)$ and $ 0 \leq F \leq 1$. It is straightforward to see that if $ F>\frac{1}{2} $, then $\rho_{AB}^{F}$ is steerable in both ways. \\

It is worth remarking on two properties of the inequality in Eq. (\ref{correlation1}). First, this inequality is asymmetric for the cases of $ \Tr{\left( \rho_A^2 \right)} \neq \Tr{\left( \rho_B^2 \right)}$. For instance, if we consider the two-qubit state in Eq. (\ref{nsinglet}), its reduced states $\rho_A^{ns}$ and $\rho_B^{ns}$ have asymmetric mixedness, i.e., $ \Tr{\left[\left(\rho_{A}^{ns}\right)^2\right]} \neq \Tr{\left[\left(\rho_{B}^{ns}\right)^2\right]}$, unless $ p =4/7,\, 2/5 $.   Second, the inequality in Eq. (\ref{correlation1}) is independent of the choice of local orthogonal observables (LOOs) in each party, because of the uniqueness of the singular values of the correlation matrix \cite{Horn2}. \\
We now make another proposition useful for detecting steerability.
\begin{propo}
If a given bipartite state $\rho_{AB}$ is nonsteerable from Alice to Bob, the following must be satisfied,
\begin{equation}
\label{Col2}
\Tr \left[ \left(C^{LOOs}\right)^T\, \left(A^{LOOs}\right)^{-1}\, C^{LOOs} \right]  \leq 1-\Tr{\left( \rho_B^2 \right)},
\end{equation} 
where $A^{-1}$ denotes the Moore-Penrose pseudoinverse of the matrix $A$ \cite{Horn2}. The term on the left-hand side of Eq. (\ref{Col2}) is independent of the choice of LOOs.  
\end{propo}
\begin{proof}
Let us consider a real symmetric matrix with a block structure as in Eq. (\ref{block}). Then the following statements are equivalent \cite{Giedke01} : (a) $ \Lambda \geq 0$, (b) $\ker{\left( \Lambda_{22} \right)} \subset \ker{\left( \Lambda_{12} \right)}$ and $\Lambda_{11}-\Lambda_{12} \Lambda_{22}^{-1} \Lambda_{12}^T$, and (c) $\ker{\left( \Lambda_{11} \right)} \subset \ker{\left( \Lambda_{12}^T \right)}$ and $ \Lambda_{22}-\Lambda_{12}^T \Lambda_{11}^{-1} \Lambda_{12} \geq 0$. \\

Applying the conditions (a) and (c) to Eqs. (\ref{BCM}) and (\ref{covariance}) leads to $B^{LOOs}-\left(C^{LOOs}\right)^T\,\left(A^{LOOs}\right)^{-1}\,C^{LOOs} \geq \kappa^{LOOs}_B \geq 0$. Since $ \Tr{\left(\kappa^{LOOs}_B\right)} = d_B-1$ and $\Tr{ \left(B^{LOOs}\right) }= d_B -\Tr{\left( \rho_B^2 \right)}$ \cite{Guhne06, Guhne07}, we obtain the desired inequality in Eq. (\ref{Col2}).  \\

It is known that given a set $\left\{ A_k \right\}$ of LOOs, any other set $ \left\{ \tilde{A}_l \right\}$ of LOOs has the form $\tilde{A}_l=\sum_{k}O_{lk} A_k $, where $O_{lk}$ is an entry of an arbitrary $ d_A^2 \times d_A^2 $ real orthogonal matrix $ O_A$ \cite{Yu05}. Similarly, there is a $ d_B^2 \times d_B^2 $ real orthogonal matrix $ O_B$ that transforms a set $\left\{ B_k \right\} $ to another set $ \left\{ \tilde{B}_l \right\} $. If we choose the sets $ \left\{ \tilde{A}_l \right\}$ and $ \left\{ \tilde{B}_l \right\} $ of LOOs instead of $\left\{ A_k \right\}$ and $\left\{ B_k \right\} $, the covariance matrix in Eq. (\ref{BCM}) is given by
\begin{equation}
\label{BCM2}
\tilde{\gamma}_{AB}=\left[ \begin{array}{cc}
\tilde{A} & \tilde{C} \\
\tilde{C}^T & \tilde{B} 
\end{array} \right] =  \left[ \begin{array}{cc}
O_A\, A^{LOOs} \, O_A^T & O_A\, C^{LOOs} \, O_B^T \\
O_B\, \left(C^{LOOs}\right)^T\, O_A^T & O_B\, B^{LOOs} \, O_B^T \\ 
\end{array} \right].
\end{equation}
Then the term on the left-hand side of Eq. (\ref{Col2}) is given by
\begin{widetext}
\begin{equation}
\begin{array}{ll}
\Tr{ \left( \tilde{C}^T\, \tilde{A}^{-1}\, \tilde{A}\right)} & =  \Tr{\left[\left( O_A C^{LOOs} O_B^T \right)^T\,\left( O_A A^{LOOs} O_A^T \right)^{-1} \left(O_A C^{LOOs} O_B^T \right)\right] } \\
                                                            & =  \Tr{\left[O_B \left(C^{LOOs}\right)^T O_A^T O_A \left(A^{LOOs}\right)^{-1} O_A^T O_A C^{LOOs} O_B^T \right] } 
                                                             =  \Tr{\left[\left(C^{LOOs}\right)^T \left(A^{LOOs}\right)^{-1} C^{LOOs} \right]},
                                                            \end{array}
\end{equation}
\end{widetext}
where we used the definition of orthogonal matrices and the permutation invariance of the trace of a matrix . 
\end{proof}
{\it Examples}- (iii) Let us consider a two-qubit Werner state given by
\begin{equation}
\label{werner}
\rho_{AB}^{W}=p|\psi^{-}\rangle \langle \psi^{-}| + \frac{\left(1-p\right)}{4} \openone \otimes \openone,
\end{equation}
where $ 0 \leq p \leq 1$. Using Eq. (\ref{Col2}) one finds that $ \rho_{AB}^{W} $ is steerable from Alice to Bob and from Bob to Alice for $ p\, >\, \frac{1}{\sqrt{3}} $. 

(iv) Let us consider a two-qutrit state represented by
\begin{equation}
\label{ex5}
\rho_{AB}^{F'}=F'| \Phi^{+} \rangle \langle \Phi^{+} | + \frac{1-F'}{2}  \left( \rho_1 + \rho_2 \right),
\end{equation}
where $ \rho_1 =1/3 \left( |01 \rangle \langle 01 | + | 12 \rangle \langle 12 | + |20 \rangle \langle 20| \right)$, $\rho_2 =1/3 \left( |02 \rangle \langle 02| + | 10 \rangle \langle 10 | + | 21 \rangle \langle 21 | \right)$, and $ 0 \leq F' \leq 1 $. We can check that if $ F' > 1/7 \left( 1+ 2\sqrt{2} \right) \approx 0.5469 $, then the state in Eq. (\ref{ex5}) is steerable in both ways. \\ \\

We now turn our attention to CV systems. 
\begin{propo}
Consider a bipartite quantum state $\rho^{CV}_{AB}$ of $ M \times N$-modes of continuous variables. If we choose local observables as $\left\{ A_k \times \openone_{B_N},\; \openone_{A_M} \otimes B_k  \right\} = \left\{ X^{(1)}_{A_k}\otimes \openone_{B_N}, P^{(1)}_{A_k}\otimes \openone_{B_N}, \openone_{A_M}\otimes X^{(1)}_{Bl}, \openone_{A_M} \otimes P^{(1)}_{B_l} \right\}$, where $ \openone_{A_M}= \openone^{\otimes M}_{A}$, $ \openone_{B_N}=\openone^{\otimes N}_{B}$, $k=1,...,M$, $l=1,...,N$ and $ A_k, B_l$ denote the $k$-th and $l$-th mode in the parties which are held by Alice and Bob, respectively. Here we omit the tensor products in each party. For instance, $ X^{(1)}_{A1} \otimes \openone_{B_N} $ means $ X^{(1)}_{A1} \otimes \openone_{A2} \otimes \cdots \openone_{AM} \otimes \openone_{B_N}$. If a state $ \rho^{CV}_{AB}$ is nonsteerable from Alice to Bob, the covariance matrix $\gamma_{AB}^{CV}$  that is constructed with those local observables has to satisfy 
\begin{equation}
\label{CVC}
\gamma_{AB}^{CV} \geq {\bf 0}_A \oplus i\Omega_B,
\end{equation} 
where $\Omega_B = \oplus_{i=1}^{N} \Omega_{Bi}$ with $\Omega_{Bi}= \frac{1}{2} \left( \begin{array}{rr} 0\;\;1 \\
-1\;\;0 \end{array} \right) $. If this inequality is violated then $ \rho^{CV}_{AB}$ is steerable from Alice to Bob.
\end{propo}

\begin{proof}
The proof of {\bf Proposition 3} is straightforward. By {\bf Theorem}, if it is impossible to steer from Alice to Bob, we have that $ \gamma^{CV}_{AB} \geq {\bf 0}_A \oplus \kappa_B^{CV} $ where $ \kappa_B^{CV}= \sum_k p_k \gamma^{CV}\left( | b_k \rangle\langle b_k | \right)$. Since $ \gamma^{CV}\left( | b_k \rangle \langle b_k | \right) $ is the covariance matrix of a physical quantum state with regular quadrature amplitudes, we obviously have $ {\bf 0}_A \oplus \kappa_B^{CV} \geq {\bf 0}_A \oplus i\Omega_B $ due to uncertainty principle \cite{Simon}. Because the sum of two positive semidefinite operators is also positive semidefinite \cite{Horn2}, we obtain the desired inequality in Eq. (\ref{CVC}).
\end{proof}
The inequality in Eq. (\ref{CVC}) is indeed the nonsteerability criterion for Gaussian states under Gaussian measurements that Wiseman {\it et al.} have derived in Ref. \cite{Wiseman07}. \\

{\it Remarks}---After completion of this work \cite{nha}, we became aware of a related interesting article \cite{adess}, which proposes moment-based steering criteria 
in a form $\Gamma_{ij}=\langle S_i^\dag S_j\rangle$. The operators $S_i$ in \cite{adess} are chosen as a product of local operators $S_i=A_k\otimes B_l$ and there arise some unobservable matrix elements that are treated as free parameters with constraints due to joint-measurability on Alice's side and quantum algebra on Bob's side. 

In our case, the correlation matrix $C$ in Eq. (\ref{BCM}) addresses all observable moments from the outset. On the other hand, the local covariance matrix $A$ at Alice site contains moments $A_{i,j}=\left( \langle A_i A_j \rangle + \langle A_j A_i \rangle \right)/2 -\langle A_i \rangle \langle A_j \rangle$ that cannot be directly determined due to the incompatibility of local measurements $\{A_i\}$. Of course, one might define a new Hermitian observable $O_A\equiv A_i A_j+A_j A_i$ to evaluate it, which is however not acceptable in a rigorous steering test where Alice is fully untrusted. Nevertheless, we have shown in the proof of Theorem that the violation of the inequality (12) is equivalent to the existence of local observables $\{{\tilde A}_k\}$ and $\{{\tilde B}_k\}$, with which one can show the violation of inequality (6). As the latter inequality requires only observable moments, the strongest test of steering becomes always possible for whatever states violating the covariance matrix criterion as a matter of principle. 

It now becomes an interesting question how those observables $\{{\tilde A}_k\}$ and $\{{\tilde B}_k\}$ can be systematically obtained for a given covariance matrix. For those states violating the inequality (16), we immediately obtain the observables $\{{\tilde A}_k\}$ and $\{{\tilde B}_k\}$ using a singular-value decomposition (SVD) of the correlation matrix $C\equiv O^{(1)}\Lambda O^{(2)}$, where $O^{(i)}$ $(i=1,2)$ are orthogonal matrices and $\Lambda$ a diagonal matrix with non-negative entries. As we have shown in Appendix, the violation of the inequality (16) is equivalent to the violation of (A. 2)---a strongest form of steering test--- by defining another LOOs ${\tilde A}_k=\sum_iO_{ik}^{(1)}A_i$ and ${\tilde B}_k=\sum_iO_{ik}^{(2)}B_i$, where $O^{(i)}$ are the same orthogonal matrices in the SVD of $C\equiv O^{(1)}\Lambda O^{(2)}$ while $\{A_i\}$ and $\{B_i\}$ are the original LOOs constructing the covariance matrix.

{\it Conclusion}---We have derived steering criteria based on covariance matrices of local observables, which can be applied to both discrete and continuous variable quantum systems. We have particularly employed local orthogonal observables (LOOs) to obtain some readily computable, and experimentally testable, criteria useful to detect steerability for quantum systems of arbitrary dimensions. We have demonstrated that these are useful to detect a wide range of entangled states particularly for two qutrits, which can be further extended to higher-dimensional systems. 
Moreover, we have shown that the Gaussian steering criteria for CV systems of $M\times N$ modes are derived as a special case of our criteria. We hope our method could be a useful tool to identify a broader set of quantum steerable states and bring a deeper understanding to quantum correlations at large.


We thank Ioannis Kogias very much for a useful discussion on our manuscript. This work is supported by an NPRP grant from Qatar National Research Fund.

\appendix*
\setcounter{equation}{0}
\section{Appendix}
\subsection{Alternative Proof of Proposition 1}

If a given bipartite quantum state $\rho_{AB}$ is nonsteerable, we know that the inequality in Eq. (\ref{gLUR}) is satisfied,
\begin{equation}
\label{gLUR2}
\sum_k \delta^2\left( g_k \hat{A}_k \otimes \openone + \openone \otimes \hat{B}_k \right) \geq C_B.
\end{equation}
Now let us consider the case that the observables $\hat{A}_k$, $\hat{B}_k$ are LOOs, i.e., $ \Tr\left({\hat{A}_k \hat{A}_l}\right)=\Tr\left({\hat{B}_k \hat{B}_l}\right)=\delta_{kl}$, and set $g_1=g_2=...=g$. Then we directly derive, from Eqs. (\ref{gLUR}) and (\ref{BLUR}), the criterion
\begin{equation}
\label{gLUR3}
\sum_{k} \delta^2\left(g\hat{A_k}\otimes\openone_B+ \openone_A \otimes \hat{B_k} \right) \geq d_B -1.
\end{equation}   
Thus any quantum state violating Eq. (\ref{gLUR3}) is steerable from Alice to Bob. Let us now set the real number $g$ to make the term on the left-hand side of Eq. (\ref{gLUR3}) as small as possible, 
\begin{equation}
\label{gg}
\begin{array}{lll}
& \frac{\partial \left( \sum_{k} \delta^2\left(g\hat{A}_k\otimes\openone_B+\openone_A \otimes \hat{B}_k \right) \right)}{\partial g } \\
& =2g\left( \sum_{k} \delta^2 \hat{A}_k \right) +2 \left( \sum_k{ \langle \hat{A}_k \otimes \hat{B}_k \rangle -\langle \hat{A}_k\rangle\langle \hat{B}_k \rangle } \right)=0 , \\
& \Rightarrow \;\; g=-\frac{\sum_{k}\langle \hat{A}_k \otimes \hat{B}_k \rangle -\langle \hat{A}_k \rangle\langle \hat{B}_k \rangle}{\sum_{k} \delta^2 \hat{A}_k}.
\end{array}
\end{equation} 
Substituting $g$ in Eq. (\ref{gg}) into the inequality (\ref{gLUR3}) and using the properites of LOOs, we obtain the simple relation for non-steerable quantum states from Alice to Bob 
\begin{equation}
\label{correl2}
\left| \sum_k{\langle \hat{A}_k\otimes \hat{B}_k \rangle - \langle \hat{A}_k \rangle \langle \hat{B}_k \rangle} \right| \leq \sqrt{\left(d_A -\Tr{\left(\rho_A^2\right)}\right) \left( 1-\Tr{\left( \rho_B^2 \right)}\right)}.
\end{equation}
The inequality in Eq. (\ref{correl2}) is always satisfied for all bipartite non-steerable states from Alice and Bob with arbitrary sets of LOOs in ${\mathcal H}_A$ and ${\mathcal H}_B$. We can now choose the LOOs that make the left-hand side of Eq. (\ref{correl2}) largest. To this aim, we use the correlation matrix $C$ ($d_A^2 \times d_B^2$ real matrix) of LOOs in Eq. (\ref{BCM}), which can be brought to a singular value decomposition with real orthogonal matrices \cite{Horn2}. The transformed LOOs by orthogonal matrices are also LOOs and we can reformulate the inequality (\ref{correl2}) as
\begin{equation}
\| C^{LOOs}\|_{tr} \leq \sqrt{\left(d_A -\Tr{\left(\rho_A^2\right)}\right) \left( 1-\Tr{\left( \rho_B^2 \right)}\right)}. 
\end{equation}

\end{document}